\documentclass[
    twocolumn,
    showpacs, notitlepage, nofootinbib,
    floatfix,
    superscriptaddress,
    prl
]{revtex4-2}

\usepackage[utf8]{inputenc}
\usepackage[english]{babel}
\usepackage[T1]{fontenc}

\usepackage{amsmath,amssymb} 
\usepackage{amsthm}
\usepackage{amsfonts}
\usepackage{physics}
\usepackage{siunitx}

\usepackage{mdframed}
\usepackage{booktabs}
\usepackage{longtable}
\usepackage{multirow}

\usepackage{graphicx}
\usepackage{xcolor}

\usepackage[linesnumbered, ruled, vlined, titlenotnumbered]{algorithm2e}

\usepackage{booktabs}

\theoremstyle{definition}

\newtheorem{lemma}{Lemma}
\newtheorem{proposition}{Proposition}
\newtheorem{theorem}{Theorem}

\newcommand{\introduce}[1]{\textit{#1}}

\newcommand{\coloration}{\ensuremath{\mathcal{C}}}

\newcommand{\dirD}{{\bf D}}
\newcommand{\dirN}{{\bf N}}
\newcommand{\dirE}{{\bf E}}
\newcommand{\dirS}{{\bf S}}
\newcommand{\dirW}{{\bf W}}
\newcommand{\dirh}{{\bf h}}
\newcommand{\dirv}{{\bf v}}

\DeclareMathOperator{\cnot}{CNOT}
\DeclareMathOperator{\HGP}{HGP}

\newcommand{\sectitle}[1]{\textit{#1 ---}}

\newenvironment{algo}[1]{
  \algorithm[H]
    \caption{#1}
    \DontPrintSemicolon
    \SetAlgoCaptionLayout{left}
    \SetAlgoHangIndent{0pt}
    \SetKwInOut{Input}{input}
    \SetKwInOut{Output}{output} 
}{
  \endalgorithm
}

\begin{document}

\title{Constant-overhead quantum error correction with thin planar connectivity}

\author{Maxime A. Tremblay}
\affiliation{
    Institut quantique
    \&
    Département de physique,
    Université de Sherbrooke,
    Sherbrooke, Qc, Canada, J1K 2R1
}

\author{Nicolas Delfosse}
\author{Michael E. Beverland}
\affiliation{
    Microsoft Quantum
    \&
    Microsoft Research,
    Redmond, WA 98052, USA
}

\begin{abstract}
Quantum LDPC codes may provide a path to build low-overhead fault-tolerant quantum computers.
However, as general LDPC codes lack geometric constraints, naïve layouts couple many distant qubits with crossing connections which could be hard to build in hardware and could result in performance-degrading crosstalk. 
We propose a 2D layout for quantum LDPC codes by decomposing their Tanner graphs into a small number of planar layers. 
Each layer contains long-range connections which do not cross.
For any CSS code with a degree-$\delta$ Tanner graph, we design stabilizer measurement circuits with depth at most $(2\delta +2)$ using at most $\lceil \delta/2 \rceil$ layers.
We observe a circuit-noise threshold of 0.28\% for a positive-rate code family using 49 physical qubits per logical qubit. 
For a physical error rate of $10^{-4}$, this family reaches a logical error rate of $10^{-15}$ using fourteen times fewer physical qubits than the surface code.
\end{abstract}

\maketitle

\begin{figure}[t]
  \includegraphics[width=.45\textwidth]{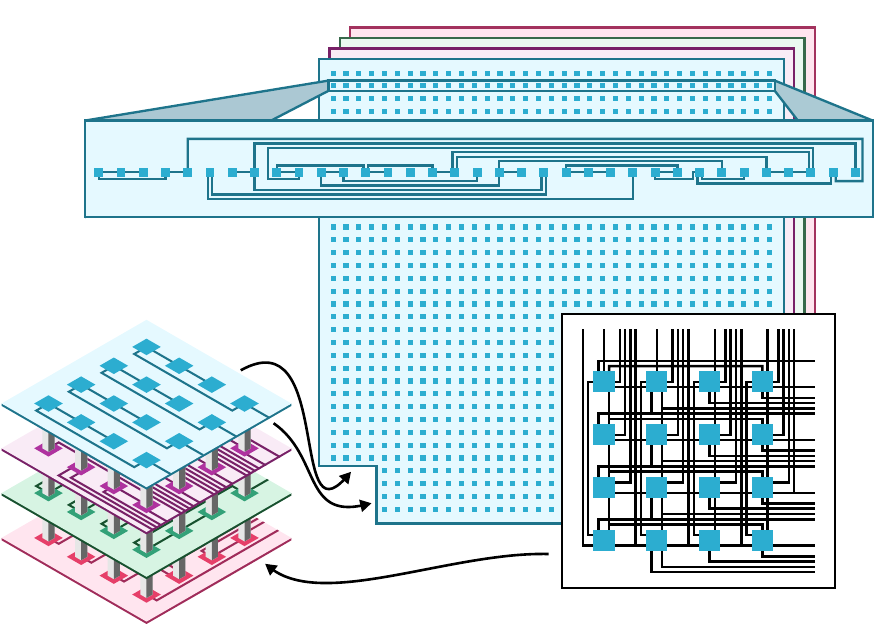}
    \caption{
        A crossing-free planar layout for hypergraph product codes with four layers.
        Stabilizer generators are measured using circuits built from single-qubit operations and CNOT gates
        between qubits connected by an edge.
        The top two layers have edges connecting qubits in the same row as shown in the zoom, 
        while the lower two layers have edges connecting qubits in the same column. 
        Each layer is planar, such that no pair of edges cross.
        For comparison,
        the lower right box shows the non-planar connections passing through
        the lower-left corner before the decomposition, which exhibits many crossings.
    }
    \label{fig::layout}
\end{figure}

Quantum error correction (QEC) is typically implemented by measuring Pauli operators called {\em stabilizer generators} of a QEC code to detect faults.
In quantum low density parity check (LDPC) codes, the stabilizer generators have low weight, making them easier to implement than general codes.
A further appealing property possessed by some quantum LDPC codes is \introduce{positive-rate}, 
allowing them to achieve arbitrarily low logical error rates with a constant ratio of tens of physical qubits per logical qubit.
For large computations, this can correspond to more than an order of magnitude lower qubit overhead than alternative codes with vanishing rate such as the surface code.
However, these positive-rate LDPC codes have non-local stabilizer generators~\cite{bravyi_tradeoffs_2010, baspin2021quantifying} making them somewhat daunting to implement in hardware.
In this work,
we seek a practical implementation of positive-rate quantum LDPC codes which performs well in a full circuit-level noise analysis.

To clarify our discussion, we define the {\em connectivity graph} of a quantum circuit, with vertices corresponding to qubits and edges connecting vertices corresponding qubits which undergo entangling operations in the circuit.
We can further define a {\em layout} as a specification of the physical locations of the connectivity graph's qubits and connections.
In this paper we are interested in circuits which measure the stabilizer generators of a QEC code. 
Given a family of codes, we focus on constant-depth stabilizer measurement circuits to avoid a build-up of errors which could spoil any fault-tolerance threshold.

A natural first question is if the non-local stabilizer generators of
positive-rate quantum LDPC codes can be measured using a circuit with
local connectivity in a 2D qubit layout.
In recent work \cite{delfosse2021bounds}, we show that although a single
non-local stabilizer can be measured in constant depth,
the full set of stabilizer generators cannot be collectively measured in constant depth 
without the ratio of logical to physical qubits vanishing. 

Given that a 2D qubit layout with local connectivity is excluded for 
positive-rate quantum LDPC codes, 
we consider quantum hardware equipped with some set of long-range connections.
With unrestricted connections, one could simply lay out the code qubits in a 2D grid 
along with an ancilla for each stabilizer generator that has connections to the code qubits in that stabilizer's support. 
However, this typically results in 
an unbounded number of crossing connections~\cite{ajtai1982crossing, leighton1983complexity}; see Fig.~\ref{fig::layout}.
In many hardware platforms, in addition to being challenging to implement,
crossing connections can spread error through crosstalk \cite{sarovar_detecting_2020, debnath_demonstration_2016, neill_blueprint_2018,ash-saki_analysis_2020}.
This raises the following question.

\medskip
\textit{
  Can we construct high-threshold constant-depth stabilizer measurement circuits for positive-rate
quantum LDPC codes in a 2D qubit layout without crossings in the two-qubit gate connectivity?
}
\medskip

In this work, we provide a positive answer to this question through the design of an
\introduce{$l$-planar layout}, consisting of qubits placed in a 2D grid, with edges separated into $l=O(1)$ planar layers, with no crossings in each layer;
see Figure~\ref{fig::layout}.
Furthermore,
we provide a circuit construction consistent with this layered planar architecture to measure the
stabilizer generators of any CSS-type~\cite{calderbank_good_1996, steane_simple_1996} quantum LDPC code in constant depth.

\begin{theorem} \label{theo:hgp_layered_layout}
  	Let $Q$ be a CSS code such that each stabilizer generator has weight at most $\delta$ and each qubit is involved in at most $\delta$ stabilizer generators.
    Then, one can implement the measurement of all the stabilizer generators of $Q$ with a circuit with depth $2\delta + 2$ using a $\left\lceil \delta/2 \right\rceil$-planar layout.
\end{theorem}

Below, after proving this theorem, we refine these results by specializing to a family of quantum codes known as hypergraph product (HGP) codes~\cite{tillich_quantum_2014} which are constructed from a pair of input graphs.
In this case, we find a low-depth stabilizer measurement circuit which reduces to the standard circuit in the case of surface codes\footnote{Surface codes are HGP codes formed when the input graphs are the Tanner graphs of a pair of repetition codes} \cite{fowler_surface_2012}. 
Moreover, we prove in Lemma~\ref{lemma:hgp_degree_bound} that the circuit depth can be reduced to $\delta + 2$ when the vertices of the input graphs of the HGP code admit a balanced ordering.

Lastly, 
we numerically explore the performance of these circuits for a family of HGP codes using the decoding routine of Grospellier and Krishna~\cite{grospellier_numerical_2019}, but replacing their idealized noise model by circuit noise.
Using our layered planar connectivity,
we obtain a circuit-noise threshold of $2.8(2)\times 10^{-3}$, providing strong 
evidence that these codes can offer a significant advantage over surface codes which have a comparable threshold.

\sectitle{Quantum LDPC codes}
All the quantum codes we consider in this work are CSS codes \cite{calderbank_good_1996, steane_simple_1996}.
Recall that a CSS code with length $n$ is defined by
a set of commuting \introduce{stabilizer generators}
$s_{X, 1}, \dots, s_{X, r_X}, s_{Z, 1}, \dots, s_{Z,r_Z}$
with $s_{X, i} \in \{I, X\}^n$ and $s_{Z, j} \in \{I, Z\}^n$.
The $X$ \introduce{Tanner graph} is the bipartite graph $T_X = (V, E)$
whose vertex set is $V = V_q \cup V_X$ where $V_q = \{q_1, \dots, q_n \}$ is the qubit set and $V_X = \qty{s_{X, 1}, \dots, s_{X, r_X}}$.
There is an edge between $q_i$ and $s_{X, j}$ 
iff $s_{X, j}$ acts non-trivially on qubit $q_i$.
The $Z$ Tanner graph $T_Z$ is defined similarly from the $Z$ stabilizer generators, and the overall Tanner graph is their union $T= T_X \cup T_Z$.
The code is a quantum LDPC code if the Tanner graph has bounded degree.
Quantum error correction works by measuring all the stabilizer generators and 
applying a correction based on the outcomes observed.
Our goal is to design practical stabilizer measurement circuits for quantum LDPC codes.

We first show that quantum circuits with a low-degree connectivity graph can be implemented with a layered planar connectivity using a small number of layers. 
This result applies to any quantum circuit (not just stabilizer measurement circuits) and in particular to all low-depth quantum circuits.

\begin{proposition}
  \label{prop::layered_connectivity}
	Let $C$ be a circuit made with single-qubit and two-qubit operations whose connectivity graph has degree at most $\delta$. 
	Then, $C$ can be implemented with a $\lceil \delta / 2 \rceil$-planar layout.
\end{proposition}

\begin{proof}
  This follows directly from the fact that for any graph with
  degree at most $\delta$,
  the smallest edge partition such that each subgraph is planar
  involves at most $\lceil \delta / 2 \rceil$ subgraphs~\cite{halton_thickness_1991, mutzel_thickness_1998}.
  Furthermore,
  since any planar graph can be drawn with arbitrary vertex location~\cite{pach_embedding_1998},
  we can fix the position of each qubit across layers.
\end{proof}

We now introduce the \introduce{coloration circuit} 
associated with an edge coloration of a Tanner graph
that can be used for any CSS code.
An \introduce{edge coloration} of a graph is a coloration of the edges
such that incident edges support distinct colors.
We will often consider a {\em minimum edge coloration}, that is an edge coloration
with a minimum number of colors.

\vspace{0.3cm}
\begin{algo}{Coloration circuit}
  \Input{A minimum edge coloration $\coloration_X$ of $T_X$.}
  \Output{The measurement outcome of all $X$ stabilizer generators.}
  Prepare an ancilla in $\ket +$ for each generator $s_{X, i}$.\;
  \For{color $c \in \coloration_X$}{
    Simultaneously apply all gates $\cnot_{i \to j}$ from the $i$th ancilla 
    to the $j$th data qubit supported on an edge $\{i, j\}$ with color~$c$.\;
  }
  Measure each ancilla in the $X$ basis.\;
\end{algo}

\begin{proposition} 
    \label{prop:css_codes_layout} 
    Let $Q$ be a CSS code with $X$ Tanner graph $T_X$.
    Then, the coloration circuit measures all the $X$ stabilizer generators of $Q$ in depth $\deg(T_X) + 2$.
\end{proposition}

\begin{proof}
The CNOTs applied in step 3 can be applied simultaneously because they correspond to edges with the same color, guaranteeing they have disjoint support.
We see the circuit measures the $X$ generators by rearranging the CNOTs, which commute, to form a sequence of single-generator measurement circuits. 
Each such circuit prepares an ancilla in $\ket{+}$, applies CNOTs from the ancilla to the generator's support, then measures the ancilla in the $X$ basis.
The depth of the coloration circuit is $\deg(T_X) + 2$ because the Tanner graph, which is bipartite, admits an edge coloration with $\deg(T_X)$ colors~\cite{alon_simple_2003}.
\end{proof}

Swapping the roles of $X$ and $Z$ provides a $Z$ stabilizer measurement circuit with depth $\deg(T_Z) + 2$. We are now equipped to prove Th.~\ref{theo:hgp_layered_layout}.

\begin{proof} [Proof of Th.~\ref{theo:hgp_layered_layout}]
By Proposition~\ref{prop:css_codes_layout}, a circuit extracting both $X$ and $Z$ syndromes with depth $\deg(T_X) + \deg(T_Z) + 2$ is formed by running the $X$ then the $Z$ circuit with two overlapping time steps.
The connectivity graph of the cardinal circuit has degree $\deg(T)$. 
Therefore Prop.~\ref{prop::layered_connectivity} proves the existence of a $\lceil \deg(T)/2 \rceil$-planar layout.
\end{proof}

For some codes, the depth can be further reduced by interleaving $X$ and $Z$ stabilizer measurements.
However, the design of an interleaved $X/Z$ stabilizer measurement circuit is non-trivial 
because the CNOT gates involved in $X$ and $Z$ measurements do not commute.
Now we specialize to HGP codes for which we provide an interleaved $X/Z$ stabilizer measurement circuit.

\begin{figure}[t]
  \includegraphics[scale=0.95]{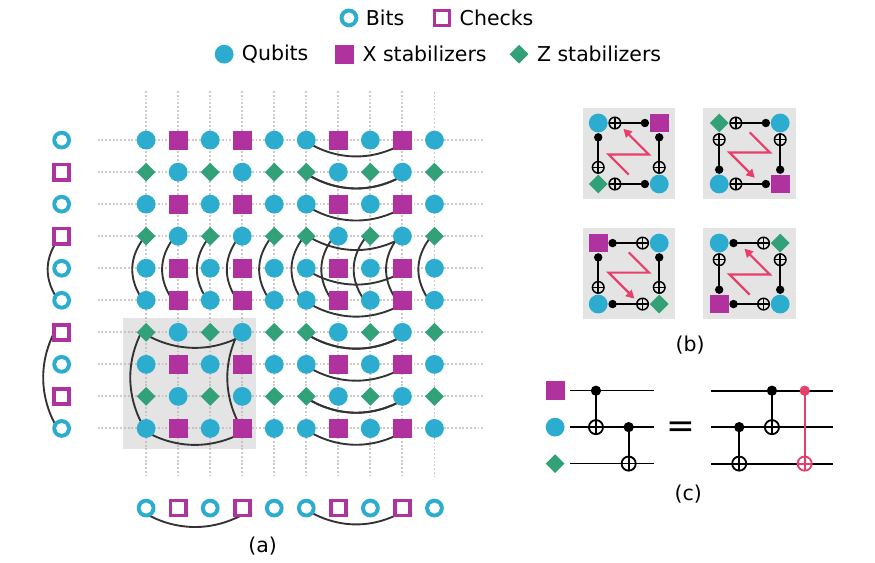}
    \caption{
      \textbf{(a)}
      The Tanner graph of a HGP code is the Cartesian product  
      of two bipartite input graphs.
      Here qubits are displayed according to their label, which does not correspond to their physical location in the 2D layout.
      \textbf{(b)}
      The four possible configurations in which an $X$ 
      and a $Z$ stabilizer can overlap.
      Red arrows show the cardinal circuit ordering.
      \textbf{(c)}
      The commutation of a pair of CNOT gates.
    }
    \label{fig::hgp_code}
\end{figure}

\sectitle{Hypergraph product codes}
The HGP code $\HGP(G_1, G_2)$~\cite{tillich_quantum_2014} is defined from the Cartesian product $G_1 \times G_2$ of two bipartite graphs~(Fig.~\ref{fig::hgp_code}(a)).
For $m \in \{1,2\}$, 
let $V_m = B_m \cup C_m$ be the vertex set of $G_m$ 
and $E_m$ be its edge set. 
Edges of $E_m$ connect a vertex of $B_m$ with a vertex of $C_m$.
We assume that each vertex of $G_m$ is given by a label $i=1, \dots, |V_m|$.
Each pair $(i, j) \in B_1 \times B_2 \cup C_1 \times C_2$
represents a data qubit of the HGP code
while each pair $(i, j)$ in $B_1 \times C_2$ (resp. $C_1 \times B_2$) 
corresponds to an $X$ (resp. $Z$) stabilizer generator.
The stabilizer generator with label $(i, j)$, denoted $S_{(i,j)}$ is supported on the qubits with label $(i', j)$ 
where $\{i', i\} \in E_1$ and $(i, j')$ with $\{j, j'\} \in E_2$.
To avoid confusion with the coordinates introduced later 
which specify the physical locations of qubits in a 2D layout, 
we refer to the pair $(i, j)$ as a \introduce{data qubit label} 
or \introduce{stabilizer generator label}.

We associate a \introduce{direction} $\dirN, \dirS, \dirE, \dirW$  with each edge.
An edge between stabilizer vertex $(i, j)$ and
qubit vertex $(i, j')$ with $j' = j + \ell \pmod{|V_1|}$
has direction $\dirN$ if $0 < \ell \leq |V_1|/2$ and direction $\dirS$ otherwise.
We define the directions $\dirE$ and $\dirW$ similarly for edges between stabilizer and qubit vertices $(i, j)$ and $(i', j)$.
For each direction $\dirD \in \{\dirN, \dirS, \dirE, \dirW\}$, 
we consider the subgraph $T_\dirD$ of the Tanner graph $T$ induced by the edges with direction~$\dirD$. 

The following circuit interleaves $X$ and $Z$ measurements and can achieve a depth lower than the coloration circuit.

\begin{algo}{Cardinal circuit}
  \label{algo::cardinal_circuit}
  \Input{A minimum edge coloration $\coloration_\dirD$ of $T_\dirD$.}
  \Output{The outcome of the measurement of all the $X$ and $Z$ stabilizer generators.}
  Prepare an ancilla in $\ket +$ for each $X$ stabilizer generator and an ancilla in $\ket 0$ for each $X$ stabilizer generator.\;
  \For{direction $\dirD \in \qty{\dirE, \dirN, \dirS, \dirW}$}{
    \For{color $c \in \coloration_\dirD$}{
      Simultaneously apply all $\cnot$ gates supported on an edge of $T_\dirD$ with color $c$.\;
    }
  }
  Measure each $X$ and $Z$ ancilla in the $X$ and $Z$ basis respectively.\;
\end{algo}
\medskip

Note that the $\cnot$ is either aligned or anti-aligned with an edge depending on the type of the stabilizer.
The control qubit of the $\cnot$ is the ancilla for $X$ stabilizers, and it is the data qubit for $Z$ stabilizers.

\begin{proposition}  \label{prop:cardinal_circuit_correctness}
	Let $Q$ be a hypergraph product code with Tanner graph $T$.
	Then, the cardinal circuit implements the measurement of all the stabilizer generators of $Q$ in depth $\deg(T_\dirN) + \deg(T_\dirS) + \deg(T_\dirE) + \deg(T_\dirW) + 2$.
\end{proposition}

\begin{proof}
	It is easy to check that the depth of the cardinal circuit is $\deg(T_\dirN) + \deg(T_\dirS) + \deg(T_\dirE) + \deg(T_\dirW) + 2$. This is because the bipartite graph $T_\dirD$ admits an edge coloration with $\deg(T_\dirD)$ colors~\cite{alon_simple_2003}.
	
 We now prove by induction that the cardinal circuit measures the stabilizer generators.    
  Denote by $C_m$ the circuit obtained by applying the cardinal circuit construction to the subset of stabilizer generators $s_1,\dots, s_m$.
  Clearly, for $m=1$, the circuit $C(s_1)$ measures the stabilizer generator $s_1$.
  We will show that the concatenation of $C_m$ and $C(s_{m+1})$, that we denote $C_m C(s_{m+1})$, has the same action as $C_{m+1}$.
    
    If all the $\cnot$ gates of $C(s_{m+1})$ commute with all the $\cnot$s of $C_m$, we can simply reorder the $\cnot$s of the circuit $C_m C(s_{m+1})$ to obtain the cardinal circuit $C_{m+1}$.
    Assume now that some $\cnot$ gates of $C(s_{m+1})$ do not commute with the gates of $C_m$.
	Again, we would like to put the $\cnot$ of $C_mC(s_{m+1})$ in the cardinal order, but swapping these $\cnot$s produces extra $\cnot$s as one can see in Fig.~\ref{fig::hgp_code}(d).
	We will show that these extra $\cnot$s cancel out.
    
  If $s_{m+1}$ is a $Z$ stabilizer, the corresponding $\cnot$s only fail to commute with $\cnot$s associated with the previous $X$ stabilizer generators $s_i$ that overlap with $s_{m+1}$.    
  Swapping these $\cnot$s produces an extra gate $\cnot_{s_i \rightarrow s_{m+1}}$ as shown in Fig.~\ref{fig::hgp_code}(c).
  By the hypergraph product construction, if $s_{m+1}$ overlaps with an $X$ stabilizer generator $s_i$, their overlap contains exactly two qubits, in one of the four possible configurations of Fig.~\ref{fig::hgp_code}(b).
 	To bring the $\cnot$s into cardinal order, we need to perform either 0 or 2 swaps between the $\cnot$s of $s_{m+1}$ and $s_i$.
	This results in two consecutive $\cnot$ gates $\cnot(s_i, s_{m+1})$ which cancel out.
  Thus, reordering the CNOT gates in $C_m C(s_{m+1})$ to produce $C_{m+1}$ preserves the action of the circuit.
  A similar argument applies when $s_{m+1}$ is an $X$ stabilizer generator.
  Applying this inductively starting with a single stabilizer generator, we reach the cardinally ordered circuit proving that it has the same action as the initial stabilizer measurement circuit.    
\end{proof}

From Prop.~\ref{prop:cardinal_circuit_correctness} we see that the cardinal circuit only has a lower depth than the coloration circuit if the Tanner subgraphs $T_D$ have sufficiently low degree.
To ensure this,
we must order the vertices $(i, j)$ of the Tanner graph to distribute the edges more equally between the four directions around each vertex.
This motivates the notion of balanced ordering that we introduce now.

A \introduce{balanced ordering} for a graph $G = (V, E)$ is a labeling of the vertices
by integers $i=1, \dots, |V|$
such that for each vertex $i$ we have $\delta_+(i) - \delta_-(i) = (\delta(i) \mod 2)$ 
where $\delta_+(i)$ is the number of vertices connected to $i$
of the form $(i+\ell) \pmod{|V|}$ with $0 < \ell \leq |V|/2$
and $\delta_-(i) = \delta(i) - \delta_+(i)$.
When the graph is not clear from the context, we will use the notation $\delta_\pm(G, i) = \delta_{\pm}(i)$.
The following lemma is proven in the Supplemental Material.

\begin{lemma} \label{lemma:hgp_degree_bound}
	Let $Q = \HGP(G_1, G_2)$ be a hypergraph product code.
	Then, we have
	$$
	\deg(T) \leq
	\sum_{\dirD = \dirN, \dirS, \dirE, \dirW} \deg(T_\dirD)
	\leq 2 \deg(T) \cdot
	$$
	Moreover, if $G_1$ and $G_2$ have only even degree vertices and admit a balanced ordering, then the lower bound is tight.
\end{lemma}

For HGP codes based on even, balanced graphs, this lemma, combined with Prop.~\ref{prop:cardinal_circuit_correctness}, proves that the depth of the cardinal circuit is about half that of the coloration circuit.

\begin{figure}[t]
    \includegraphics[scale=0.9]{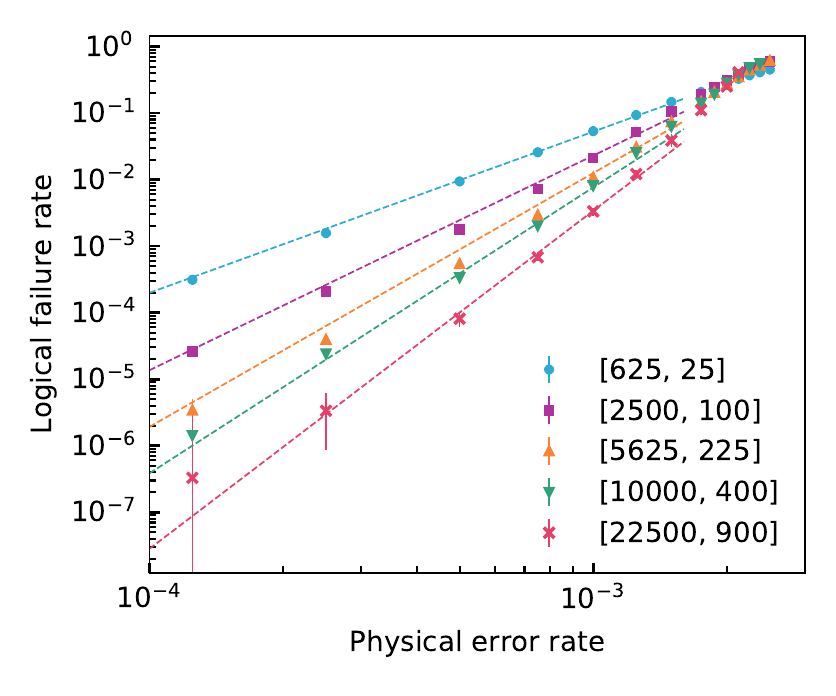}
    \caption{
      The failure rate per round averaged over 10 successive rounds of error correction.
      The dashed lines are obtained using $P_L(p, k) = c_1 \qty(p/p_t)^{c_2 k^{c_3}}$ finding fitting constants $c_1=0.64$, $c_2=1.3$, $c_3=0.21$ and the threshold $p_t = 2.8(2) \times 10^{-3}$.
    }
    \label{fig::results}
\end{figure}

\sectitle{Numerical results}
We use the standard circuit noise model.
We simulate the performance of a family of HGP codes with parameters $[[25s^2, s^2]]$, {\em i.e} encoding $k = s^2$ logical qubits into $n = 25s^2$ physical qubits.
The syndrome extraction is performed with the cardinal circuit using $24s^2$ ancilla qubits.
More details about the noise model and the codes can be found in the Supplementary Material.

To decode and correct faults,
we use the approach set out in Ref.~\cite{grospellier_combining_2020},
which is based on Belief Propagation (BP)~\cite{gallager_low-density_1962, poulin_iterative_2008}
and Small Set Flip (SSF)~\cite{leverrier_quantum_2015}.
After each round of stabilizer extraction, 
BP is iterated until there is a local minimum of the number of violated stabilizer measurements.
To probe the overall performance after $T$ rounds of error correction,
a perfect stabilizer measurement round is applied and BP and SSF are alternated until the SSF decoder converges to a correction.
If either no correction is found, or if the net effect of all noise and corrections is to introduce a logical operator, we say that a failure has occurred by round $T$.

By sampling over many realizations of this procedure with outcomes extracted using the cardinal circuit,
we estimate the logical failure rate of a family of HGP codes over a range of physical failure rates as shown in Fig.~\ref{fig::results}.
We use a simple extrapolation of the data to estimate a threshold of $p_{t} = 2.8(2) \times 10^{-3}$ and compare the qubit overhead with that of the surface code in Table~\ref{tab:data}.
Based on our simulation, we propose the heuristic formula
\begin{align} \label{eq:heuristic_for_pLog}
P_L(p, k) = c_1 \qty(p/p_t)^{c_2 k^{c_3}}
\end{align}
where $c_1=0.64$, $c_2=1.3$, $c_3=0.21$ and $k=24 s^2$ to provide an estimate of the logical failure rate per round with physical error rate $p < p_{t}$.
Further details on the code construction and our numerical approach can be found in the Supplementary Material.

\begin{table} 
    \begin{center}
        \begin{tabular} { l c c c }
            \toprule
           	Logical failure rate & $10^{-9}$ & $10^{-12}$ & $10^{-15}$ \\
            \midrule
            Logical qubits & \num{1600} & \num{6400} & \num{18496} \\
            Surface code physical qubits & \num{387200} & \num{2880000}  & \num{13354112} \\
            HGP code physical qubits & \num{78400}   & \num{313600}   & \num{906304} \\
            Improvement using HGP codes & \num{4.94}$\times$   & \num{9.18}$\times$  & \num{14.73}$\times$ \\
            \bottomrule
        \end{tabular}
    \end{center}
    \caption{Total number of data and ancilla qubits required to achieve specific logical failure rates with a physical error rate of $10^{-4}$ using surface codes and HGP codes.
	The HGP code data is estimated from Eq.~\ref{eq:heuristic_for_pLog}, while the surface code data is estimated using the formula $P'_L(p, k, d) = a k \qty(p/p'_t)^{(d+1)/2}$ with optimistic values of $p'_t = 0.011$ and $a = 0.03$ from Ref.~\cite{fowler_surface_2012} and Ref.~\cite{wang_surface_2011}.    
    }
    \label{tab:data}
\end{table}

\sectitle{Outlook and hardware challenges}
We have shown that syndrome extraction can be implemented in constant depth using a planar layout for any CSS quantum LDPC code, including HGP codes, but also many other families of interest including hyperbolic codes~\cite{breuckmann2017hyperbolic} 
and homological product codes~\cite{bravyi2014homological}.
Our design simultaneously seeks to minimize the depth of the stabilizer measurement circuit which achieves faster quantum error correction and reduce the time allowed for error buildup, while also avoiding crossings of the connections that couple qubits which is expected to improve fabrication and reduce cross-talk. 
Further improvements could come from optimizing the circuit to minimize the spreading of errors~\cite{breuckmann2021ldpc} using better decoders~\cite{panteleev2019degenerate, roffe2020decoding, delfosse2021toward, quintavalle2021lifting}, or by leveraging improved planar graph algorithms.

We hope that these significant quantum error-correction advantages will motivate experimental teams to overcome the challenges to build quantum hardware in planar layouts.
We foresee two major obstacles.
Firstly, our design requires a number of long range links within each layer. 
Significant experimental progress has been made in that direction using for instance photonic couplings to establish long range connections~\cite{trifunovic_long-distance_2013, monroe2014large, monroe2014large, nickerson2014freely, tosi_silicon_2017, ho_ultrafast_2019, morse2017photonic, bergeron2020silicon, metelmann2018nonreciprocal} but it is unclear which of these approaches could be scaled to larger systems.
Secondly, there is a tension between the need for insulation between the layers to reduce crosstalk and the fact that data qubits must participate in all layers.

\begin{acknowledgments}
\sectitle{Acknowledgments}
The authors would like to thank David Poulin for his encouragements in the early stage of this project, 
Jeongwan Haah for his comments on a preliminary version of this work
and Nouédyn Baspin for insightful discussions.
\end{acknowledgments}

\section{Supplemental Material}

\subsection{Tanner subgraphs for HGP codes}

Here we prove the bound stated in Lemma~\ref{lemma:hgp_degree_bound} on the degree of sub graphs of the Tanner graphs of HGP codes.

\begin{proof}
Edges of the Tanner graph of the form $\{(i, j), (i', j)\}$ 
are called \introduce{horizontal edges} 
and the edges $\{(i, j), (i, j')\}$ are \introduce{vertical edges}.
The lower bound is trivial because the four types of edges form a partition of the edge set of $T$.
Consider the subgraph $T_{\dirh}$ (resp. $T_{\dirv}$) of $T$ induced by horizontal (resp. vertical) edges.
By definition, $T_\dirE$ and $T_\dirW$ are subgraphs of $T_\dirh$ which implies $\deg(T_\dirE), \deg(T_\dirW) \leq \deg(T_\dirh)$.
Similarly, we obtain $\deg(T_\dirS), \deg(T_\dirN) \leq \deg(T_\dirv)$.
Moreover, due the product structure of the graph, we have 
$\deg(T) = \deg(T_\dirh) + \deg(T_\dirv)$, which leads to the upper bound.

Assume that $G_1$ and $G_2$ have only even degree vertices and are equipped with a balanced ordering and let us prove that the lower bound is tight.
Then, we have $\delta_\pm(G_m, i) = \deg_{G_m}(i)/2$.
Let $(i, j)$ be a maximum degree vertex in $T_\dirN$.
If $(i, j)$ is a stabilizer vertex, then by definition of the direction $\dirN$, the degree of $(i, j)$ in the graph $T_\dirN$ is given by $\delta_+(G_2, j)$.
If $(i, j)$ is a qubit vertex, its degree is given by $\delta_-(G_2, j)$.
Therefore, we always have
\begin{align*}
\deg(T_\dirN) 
	& \leq \max( \delta_+(G_2, j), \delta_-(G_2, j) ) \\
	& = \deg_{G_2}(j)/2 \\
	& \leq \deg(G_2) / 2.
\end{align*}
By the same reasoning, we obtain 
$
\deg(T_\dirS) \leq \deg(G_2) / 2
$
and combining both equations, we find 
\begin{align} \label{eq:lemma_proof_bound_G2}
\deg(T_\dirS) + \deg(T_\dirN) \leq \deg(G_2).
\end{align}
Applying the same argument to horizontal edges, we get
\begin{align} \label{eq:lemma_proof_bound_G1}
\deg(T_\dirE) + \deg(T_\dirW) \leq \deg(G_1).
\end{align}
Using Eq.~\eqref{eq:lemma_proof_bound_G2} and \eqref{eq:lemma_proof_bound_G1}, we obtain
\begin{align*}
\sum_{\dirD = \dirN, \dirS, \dirE, \dirW} \deg(T_\dirD)
	& \leq \deg(G_1) + \deg(G_2)
\end{align*}
which concludes the proof of the lemma because $\deg(G_1) + \deg(G_2) = \deg(T)$.
\end{proof}

\subsection{Numerical simulations}

In our numerical studies we consider HGP codes constructed from the product of random (3, 4)-regular Tanner graphs with $4s$ bit nodes and $3s$ check nodes with girth at least 8.
The codes are generated using the procedure of Ref.~\cite{grospellier_combining_2020} and for each $s$, 
we pick the best code from a few hundred samples by decoding i.i.d noise with
the SSF decoder.
The actual number of samples is based on the decoding time and the computing resources.
Depending on the code length, this ranges from 50 to 1000 samples.
For all these graphs,
we found a balanced vertex ordering
using a greedy Metropolis-Hasting inspired algorithm by defining a measure of how close to balanced is a vertex ordering of a graph.

We use a circuit noise model in which each operation, including an identity gate, is faulty with probability $p$.
If a unitary operation is faulty, we apply a random uniform non-trivial Pauli error on the support of the operation.
If a single-qubit measurement is faulty, its outcome is flipped.

In our simulations, we estimate the logical failure rate by averaging over $T$ successive rounds of error correction.
Based on Fig.~\ref{fig::failure_rate_per_round} we can assume that the logical error rate after $T$ rounds behaves roughly like $P_L(T) = c + q T$ for some constants $c$ and $q$ that depend on the noise rate and the code. 
For simplicity, we estimate the logical error rate per round as $P_L = P_L(T) / T$ for $T = 10$ rounds of error correction. 
This value leads to a pessimistic estimate of the logical failure rate after $T$ rounds for any $T \geq 10$. 
This is a reasonable model because we expect these code blocks to be kept alive for many logical cycles during a quantum computation because they encode hundreds of thousands of logical qubits. Indeed, one could conceivably run a full large scale quantum algorithm using only the logical qubits encoded in a single block of the codes simulated in Fig.~\ref{fig::results}.

\begin{figure}[t]
  \includegraphics{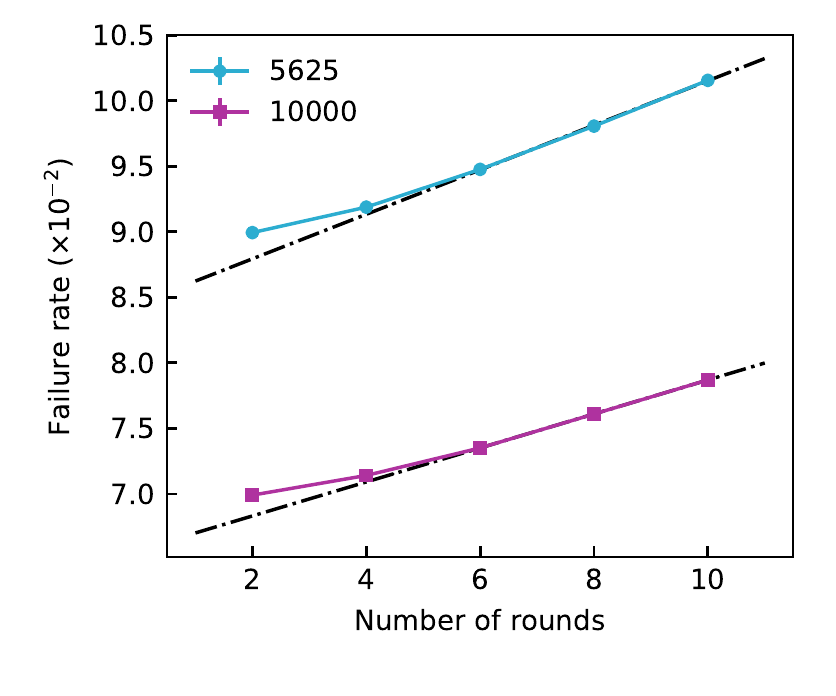}
    \caption{
		Failure rate after an increasing number of error correcting rounds
    for codes of length 5625 and 10000 with $p = 10^{-3}$.
		We observe that after 6 rounds,
		the increment per round is roughly constant.
    }
    \label{fig::failure_rate_per_round}
\end{figure}

%
%

%

\end{document}